\documentclass{llncs}

\usepackage[english]{babel}
\usepackage{german}
\usepackage{amssymb}
\usepackage{amsfonts}
\usepackage{amsmath}
\usepackage[latin1]{inputenc}
\usepackage{graphicx}
\usepackage{subfigure}
\usepackage{algorithm}
\usepackage[noend]{algorithmic}
\usepackage{amsmath,amssymb,cite}
\usepackage{lineno,footmisc,marvosym}

\newtheorem{observation}{Observation}
\begin{document}

\title{Minimum Bisection is NP-hard \\
on Unit Disk Graphs}
\author{Josep D{\'i}az\inst{1} \and George B. Mertzios\inst{2}\thanks{%
Partially supported by the EPSRC Grant~EP/K022660/1.}}
\institute{Departament de Llenguatges i Sistemes Inform{\'a}tics, \\
Universitat Polit{\'e}cnica de Catalunya, Spain.
\and School of Engineering and Computing Sciences, Durham University, UK. \\ 
\texttt{diaz@lsi.upc.edu, george.mertzios@durham.ac.uk}}
\maketitle

\begin{abstract}
In this paper we prove that the \textsc{Min-Bisection} problem is NP-hard on 
\emph{unit disk graphs}, thus solving a longstanding open question.\newline

\end{abstract}

\section{Introduction}

\label{intro-sec}

The problem of appropriately partitioning the vertices of a given graph into
subsets, such that certain conditions are fulfilled, is a fundamental
algorithmic problem. Apart from their evident theoretical interest, graph
partitioning problems have great practical relevance in a wide spectrum of
applications, such as in computer vision, image processing, and VLSI layout
design, among others, as they appear in many divide-and-conquer algorithms
(for an overview see~\cite{GraphPartitioning-Book}). In particular, the
problem of partitioning a graph into equal sized components, while
minimizing the number of edges among the components turns out to be very
important in parallel computing. For instance, to parallelize applications
we usually need to evenly distribute the computational load to processors,
while minimizing the communication between processors.

Given a simple graph $G=(V,E)$ and ${k\geq 2}$, a \emph{balanced }$k$\emph{%
-partition} of $G=(V,E)$ is a partition of $V$ into $k$ vertex sets $%
V_{1},V_{2},\ldots ,V_{k}$ such that $|V_{i}|\leq \left\lceil \frac{|V|}{k}%
\right\rceil$ for every $i=1,2,\ldots ,k$. The \emph{cut size} (or simply,
the \emph{size}) of a balanced $k$-partition is the number of edges of $G$
with one endpoint in a set $V_{i}$ and the other endpoint in a set $V_{j}$,
where $i\neq j$. In particular, for $k=2$, a balanced $2$-partition of $G$
is also termed a \emph{bisection} of $G$. The \emph{minimum bisection}
problem (or simply, \textsc{Min-Bisection}) is the problem, given a graph $G$%
, to compute a bisection of~$G$ with the minimum possible size, also known
as the \emph{bisection width} of $G$.

Due to the practical importance of \textsc{Min-Bisection}, several
heuristics and exact algorithms have been developed, which are quite
efficient in practice~\cite{GraphPartitioning-Book}, from the first ones in
the 70's~\cite{KL70} up to the very efficient one described in~\cite%
{DellingGRW12}. However, from the theoretical viewpoint, \textsc{%
Min-Bisection} has been one of the most intriguing problems in algorithmic
graph theory so far. This problem is well known to be NP-hard for general
graphs~\cite{GareyJohnson}, while it remains NP-hard when restricted to the
class of everywhere dense graphs~\cite{MacGregor78} (i.e.~graphs with
minimum degree~$\Omega (n)$), to the class of bounded maximum degree graphs~%
\cite{MacGregor78}, or to the class of $d$-regular graphs~\cite{Bui87}. On
the positive side, very recently it has been proved that \textsc{%
Min-Bisection} is fixed parameter tractable~\cite{CyganLPPS14}, while the
currently best known approximation ratio is $O(\log n)$~\cite{Racke08}.
Furthermore, it is known that \textsc{Min-Bisection} can be solved in
polynomial time on trees and hypercubes~\cite{MacGregor78,DiazPS02}, on
graphs with bounded treewidth~\cite{JansenKLS05}, as well as on grid graphs
with a constant number of holes~\cite{PapadimitriouSideri96,FeldmannW11}.

In spite of this, the complexity status of \textsc{Min-Bisection} on planar
graphs, on grid graphs with an arbitrary number of holes, and on unit disk
graphs have remained longstanding open problems so far~\cite%
{FeldmannW11,DiazPPS01,Karpinski02,Kahruman09}. The first two of these
problems are equivalent, as there exists a polynomial time reduction from
planar graphs to grid graphs with holes~\cite{PapadimitriouSideri96}.
Furthermore, there exists a polynomial time reduction from planar graphs
with maximum degree $4$ to unit disk graphs~\cite{DiazPPS01}. Therefore,
since grid graphs with holes are planar graphs of maximum degree $4$, there
exists a polynomial reduction of \textsc{Min-Bisection} from planar graphs
to unit disk graphs. Another motivation for studying \textsc{Min-Bisection}
on unit disk graphs comes from the area of wireless communication networks~%
\cite{AkyildizSSC02,BradonjicEFSS10}, as the bisection width determines the
communication bandwidth of the network~\cite{Hromkovic05}.

\vspace{0.2cm}

\noindent \textbf{Our contribution.} In this paper we resolve the complexity
of \textsc{Min-Bisection} on unit disk graphs. In particular, we prove that
this problem is NP-hard by providing a polynomial reduction from a variant
of the maximum satisfiability problem, namely from the \emph{monotone
Max-XOR($3$)} problem. This optimization problem (which is also known as the
monotone Max-$2$-XOR($3$) problem) essentially encodes the \emph{Max-Cut}
problem on $3$-regular graphs. Consider a monotone XOR-boolean formula $\phi 
$ with variables $x_{1},x_{2},\ldots ,x_{n}$, i.e.~a boolean formula that is
the conjunction of XOR-clauses of the form $(x_{i}\oplus x_{k})$, where no
variable is negated. If, in addition, every variable $x_{i}$ appears in
exactly $k$ XOR-clauses in $\phi $, then $\phi $ is called a \emph{monotone
XOR(}$k$\emph{)} formula. The \emph{monotone Max-XOR(}$k$\emph{)} problem
is, given a monotone XOR($k$) formula $\phi $, to compute a truth assignment
of the variables $x_{1},x_{2},\ldots ,x_{n}$ that XOR-satisfies the largest
possible number of clauses of $\phi $. Recall here that the clause $%
(x_{i}\oplus x_{k})$ is XOR-satisfied by a truth assignment $\tau $ if and
only if $x_{i}\neq x_{k}$ in $\tau $. Given a monotone XOR($k$) formula $%
\phi $, we construct a unit disk graph $H_{\phi }$ such that the truth
assignments that XOR-satisfy the maximum number of clauses in $\phi $
correspond bijectively to the minimum bisections in $H_{\phi }$, thus
proving that \textsc{Min-Bisection} is NP-hard on unit disk graphs.

\vspace{0.2cm}

\noindent \textbf{Organization of the paper.} Necessary definitions and
notation are given in Section~\ref{preliminaries-sec}. In Section~\ref%
{Gn-sec}, given a monotone XOR($3$)-formula $\phi $ with $n$ variables, we
construct an auxiliary unit disk graph $G_{n}$, which depends only on the
size $n$ of $\phi $ (and not on $\phi $ itself). In Section~\ref{H-phi-sec}
we present our reduction from the monotone Max-XOR($3$) problem to \textsc{%
Min-Bisection} on unit disk graphs, by modifying the graph $G_{n}$ to a unit
disk graph $H_{\phi }$ which also depends on the formula $\phi $ itself.
Finally we discuss the presented results and remaining open problems in
Section~\ref{conclusions}.

\section{Preliminaries and Notation\label{preliminaries-sec}}

We consider in this article simple undirected graphs with no loops or
multiple edges. In an undirected graph $G=(V,E)$, the edge between vertices $%
u$ and $v$ is denoted by~$uv$, and in this case $u$ and $v$ are said to be 
\emph{adjacent} in $G$. For every vertex $u\in V$ \ the \emph{neighborhood}
of $u$ is the set $N(u)=\{v\in V\ |\ uv\in E\}$ of its adjacent vertices and
its \emph{closed neighborhood} is $N[u]=N(u)\cup \{u\}$. The subgraph of $G$
that is \emph{induced} by the vertex subset $S\subseteq V$ is denoted $G[S]$%
. Furthermore a vertex subset $S\subseteq V$ induces a \emph{clique} in $G$
if $uv\in E$ for every pair $u,v\in S$.

A graph $G=(V,E)$ with $n$ vertices is the \emph{intersection graph} of a
family ${F=\{S_{1},\ldots ,S_{n}\}}$ of subsets of a set $S$ if there exists
a bijection $\mu :V\rightarrow F$ such that for any two distinct vertices~${%
u,v\in V}$, $uv\in E$ if and only if $\mu (u)\cap \mu (v)\neq \emptyset $.
Then, $F$ is called an \emph{intersection model} of $G$. A graph $G$ is a 
\emph{disk} graph if $G$ is the intersection graph of a set of disks
(i.e.~circles together with their internal area) in the plane. A disk graph $%
G$ is a \emph{unit disk} graph if there exists a disk intersection model for 
$G$ where all disks have equal radius (without loss of generality, all their
radii are equal to $1$). Given a disk (resp.~unit disk) graph $G$, an
intersection model of $G$ with disks (resp.~unit disks) in the plane is
called a \emph{disk} (resp.~\emph{unit disk}) \emph{representation} of $G$.
Alternatively, unit disk graphs can be defined as the graphs that can be
represented by a set of points on the plane (where every point corresponds
to a vertex) such that two vertices intersect if and only if the
corresponding points lie at a distance at most some fixed constant $c$ (for
example $c=1$). Although these two definitions of unit disk graphs are
equivalent, in this paper we use the representation with the unit disks
instead of the representation with the points.

Note that any unit disk representation $R$ of a unit disk graph $G=(V,E)$
can be completely described by specifying the centers $c_{v}$ of the unit
disks $D_{v}$, where $v\in V$, while for any disk representation we also
need to specify the radius $r_{v}$ of every disk $D_{v}$, $v\in V$. Given a
graph $G$, it is NP-hard to decide whether $G$ is a disk (resp.~unit disk)
graph~\cite{Kratochvil96,Breu98}. Given a unit disk representation $R$ of a
unit disk graph $G$, in the remainder of the paper we may not distinguish
for simplicity between a vertex of $G$ and the corresponding unit disk in $R$%
, whenever it is clear from the context. It is well known that the Max-XOR
problem is NP-hard. Furthermore, it remains NP-hard even if the given
formula $\phi $ is restricted to be a monotone XOR($3$) formula. For the
sake of completeness we provide in the next lemma a proof of this fact.

\begin{lemma}
\label{xor-np-hard-lem}Monotone Max-XOR($3$) is NP-hard.
\end{lemma}

\begin{proof}
The Max-Cut problem is NP-hard, even when restricted to cubic graphs, i.e.
to graphs $G=(V,E)$ where $|N(u)|=3$ for every vertex $u\in V$ \cite%
{Yannakakis78}. Consider a cubic graph $G=(V,E)$. We construct from $G$ a
monotone XOR($3$) formula $\phi $ as follows. First we define a boolean
variable $x_{u}$ for every vertex $u\in V$. Furthermore for every edge $%
uv\in E$ we define the XOR-clause $(x_{u}\oplus x_{v})$ and we define $\phi $
to be the conjunction of all these clauses. Then, $G$ has a $2$-partition
(i.e.~a cut) of size $k$ if and only if there exists a satisfying assignment 
$\tau $ of $\phi $ that XOR-satisfies $k$ clauses of $\phi $. Indeed, for
the first direction, consider such a $2$-partition of $G$ into sets $V_{1}$
and $V_{2}$ with $k$ edges between $V_{1}$ and $V_{2}$, and define the truth
assignment $\tau $ such that $x_{u}=1$ if $u\in V_{1}$ and $x_{u}=0$ if $%
u\in V_{2}$. Then $\tau $ satisfies $k$ clauses of $\phi $. For the opposite
direction, consider a truth assignment $\tau $ that XOR-satisfies $k$
clauses of $\phi $ and define a $2$-partition of $G$ into sets $V_{1}$ and $%
V_{2}$ such that $u\in V_{1}$ if $x_{u}=1$ and $u\in V_{2}$ if $x_{u}=0$.
Then this $2$-partition has size $k$. This completes the proof of the lemma.%
\qed
\end{proof}

\section{Construction of the unit disk graph $G_{n}$\label{Gn-sec}}

In this section we present the construction of the auxiliary unit disk graph 
$G_{n}$, given a monotone XOR($3$)-formula $\phi $ with $n$ variables. Note
that $G_{n}$ depends only on the size of the formula $\phi $ and not on $%
\phi $ itself. Using this auxiliary graph $G_{n}$ we will then construct in
Section~\ref{H-phi-sec} the unit disk graph $H_{\phi }$, which depends also
on $\phi $ itself, completing thus the NP-hardness reduction from monotone
Max-XOR($3$) to the minimum bisection problem on unit disk graphs.

We define $G_{n}$ by providing a unit disk representation $R_{n}$ for it.
For simplicity of the presentation of this construction, we first define a
set of halflines on the plane, on which all centers of the disks are located
in the representation~$R_{n}$.

\subsection{The half-lines containing the disk centers\label{tracks-subsec}}

Denote the variables of the formula $\phi $ by $\{x_{1},x_{2},\ldots
,x_{n}\} $. Define for simplicity the values $d_{1}=5.6$ and $d_{2}=7.2$.
For every variable $x_{i}$, where $i\in \{1,2,\ldots ,n\}$, we define the
following four points in the plane:

\begin{itemize}
\item ${p_{i,0}=(2id_{1},2(i-1)d_{2})}$ and ${%
p_{i,1}=((2i-1)d_{1},(2i-1)d_{2})}$, which are called the \emph{bend points}
for variable $x_{i}$, and

\item ${q_{i,0}=((2i-1)d_{1},2(i-1)d_{2})}$ and ${%
r_{i,0}=(2id_{1},(2i-1)d_{2})}$, which are called the \emph{auxiliary points}
for variable~$x_{i}$.
\end{itemize}

Then, starting from point $p_{i,j}$, where $i\in \{1,2,\ldots ,n\}$ and $%
j\in \{0,1\}$, we draw in the plane one halfline parallel to the $x$-axis
pointing to the left and one halfline parallel to the $y$-axis pointing
upwards. The union of these two halflines on the plane is called the \emph{%
track} $T_{i,j}$ of point $p_{i,j}$. Note that, by definition of the points $%
p_{i,j}$, the tracks $T_{i,0}$ and $T_{i,1}$ do not have any common point,
and that, whenever $i\neq k$, the tracks ${T_{i,j}}$ and ${T_{k,\ell }}$
have exactly one common point. Furthermore note that, for every $i\in
\{1,2,\ldots ,n\}$, both auxiliary points $q_{i,0}$ and $r_{i,0}$ belong to
the track $T_{i,0}$. The construction of the tracks is illustrated in Figure~%
\ref{track-fig}.

\begin{figure}[h]
\centering
\includegraphics[width=0.5\textwidth]{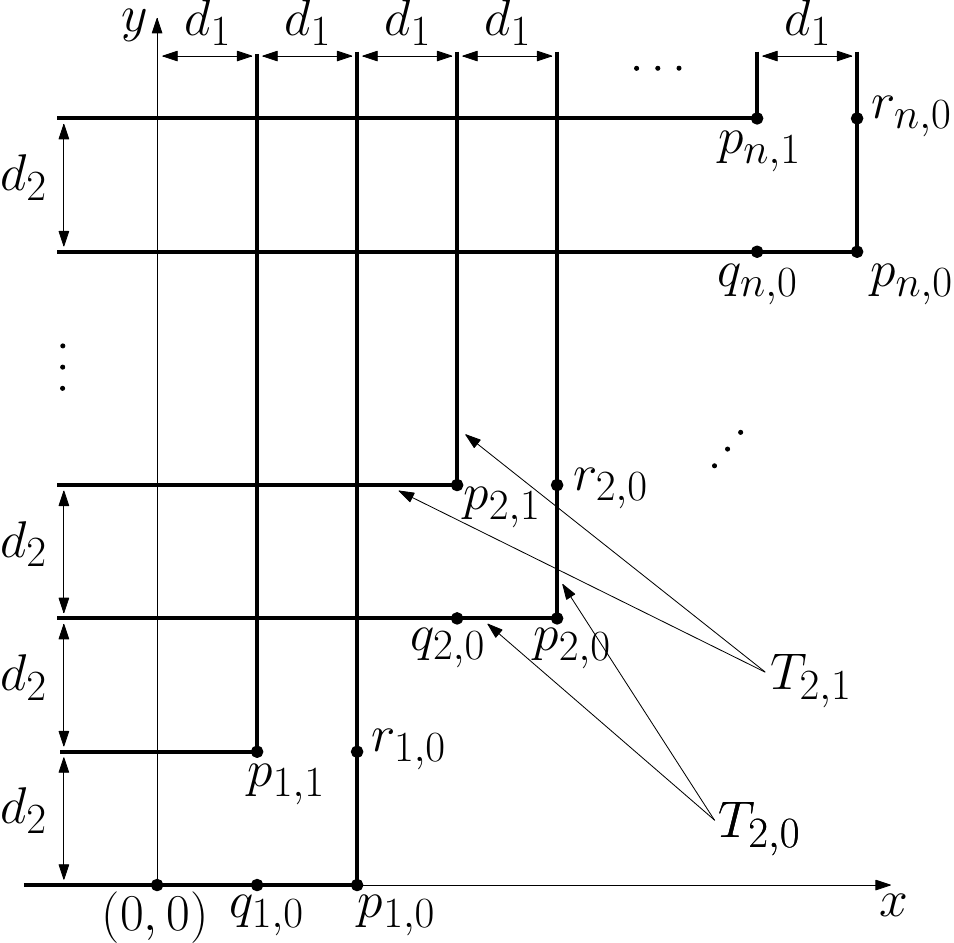}
\caption{The construction of the points~$p_{i,j}$ and the tracks~$T_{i,j}$,
where~${1\leq i\leq n}$ and~${j\in \{0,1\}}$.}
\label{track-fig}
\end{figure}

We will construct the unit disk representation $R_{n}$ of the graph $G_{n}$
in such a way that the union of all tracks $T_{i,j}$ will contain the
centers of all disks in $R_{n}$.The construction of $R_{n}$ is done by
repeatedly placing on the tracks $T_{i,j}$ (cf.~Figure~\ref{track-fig})
multiple copies of three particular unit disk representations~$Q_{1}(p)$,~$%
Q_{2}(p)$,~and $Q_{3}(p)$ (each of them including $2n^{6}+2$ unit disks),
which we use as gadgets in our construction. Before we define these gadgets
we need to define first the notion of a $(t,p)$-crowd.

\begin{definition}
\label{crowd-def}Let $\varepsilon >0$ be infinitesimally small. Let $t\geq 1$
and ${p=(p}_{x}{,p}_{y}{)}$ be a point in the plane. Then, the \emph{%
horizontal }$(t,p)$\emph{-crowd} (resp.~the \emph{vertical }$(t,p)$\emph{%
-crowd}) is a set of $t$ unit disks whose centers are equally distributed
between the points ${(p}_{x}{-\varepsilon ,p}_{y}{)}$ and ${(p}_{x}{%
+\varepsilon ,p}_{y}{)}$ (resp.~between the points ${(p}_{x}{,p}_{y}{%
-\varepsilon )}$ and~${(p}_{x}{,p}_{y}{+\varepsilon )}$).
\end{definition}

Note that, by Definition~\ref{crowd-def}, both the horizontal and the
vertical $(t,p)$-crowds represent a clique of $t$ vertices. Furthermore note
that both the horizontal and the vertical $(1,p)$-crowds consist of a single
unit disk centered at point $p$. For simplicity of the presentation, we will
graphically depict in the following a $(t,p)$-crowd just by a disk with a
dashed contour centered at point $p$, and having the number $t$ written next
to it cf.~Figure~\ref{crowd-fig}. Furthermore, whenever the point $p$ lies
on the horizontal (resp.~vertical) halfline of a track $T_{i,j}$, then any $%
(t,p)$-crowd will be meant to be a horizontal (resp.~vertical) $(t,p)$%
-crowd. For instance, a horizontal $(t,p)$-crowd is illustrated in Figure~%
\ref{crowd-fig}.

\begin{figure}[tbh]
\centering
\includegraphics[scale=0.68]{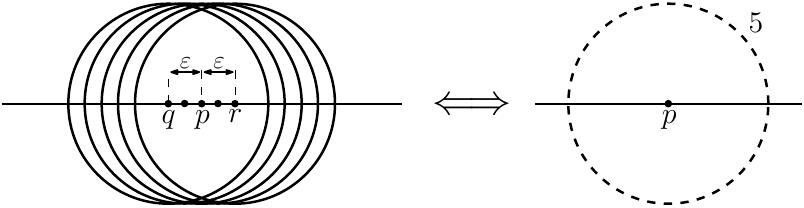}
\caption{A horizontal $(t,p)$-crowd and an equivalent way to represent it
using a disk with a dashed contour centered at point $p$, where $t=5$, $%
\protect\varepsilon>0$ is infinitesimally small, $p=(p_{x},p_{y})$, $%
q=(p_{x}-\protect\varepsilon,p_{y})$, and $r=(p_{x}+\protect\varepsilon%
,p_{y})$.}
\label{crowd-fig}
\end{figure}

\subsection{Three useful gadgets\label{gadgets-subsec}}

Let $p=(p_{x},p_{y})$ be a point on a track $T_{i,j}$. Whenever $p$ lies on
the horizontal halfline of $T_{i,j}$, we define for any $\delta >0$ (with a
slight abuse of notation) the points $p-\delta =(p_{x}-\delta ,p_{y})$ and $%
p+\delta =(p_{x}+\delta ,p_{y})$. Similarly, whenever $p$ lies on the
vertical halfline of $T_{i,j}$, we define for any $\delta >0$ the points $%
p-\delta =(p_{x},p_{y}-\delta )$ and $p+\delta =(p_{x},p_{y}+\delta )$.
Assume first that $p$ lies on the \emph{horizontal} halfline of~$T_{i,j}$.
Then we define the unit disk representation $Q_{1}(p)$ as follows:

\begin{itemize}
\item $Q_{1}(p)$ consists of the horizontal ${(n^{3},p+0.9)}$-crowd, the
horizontal ${(2n^{6}-2n^{3}+2,p+2.8)}$-crowd, and the horizontal ${%
(n^{3},p+4.7)}$-crowd, as it is illustrated in Figure~\ref{gadget-fig-1}.
\end{itemize}

Assume now that $p$ lies on the \emph{vertical} halfline of $T_{i,j}$, we
define the unit disk representations $Q_{2}(p)$ and $Q_{3}(p)$ as follows:

\begin{itemize}
\item $Q_{2}(p)$ consists of a single unit disk centered at point $p$, the
vertical ${(n^{6},p+1.8)}$-crowd, a single unit disk centered at point $%
p+3.6 $, and the vertical ${(n^{6},p+5.4)}$-crowd, as it is illustrated in
Figure~\ref{gadget-fig-2}.

\item $Q_{3}(p)$ consists of a single unit disk centered at point $p$, the
vertical ${(n^{6},p+1.7)}$-crowd, a single unit disk centered at point $%
p+3.6 $, and the vertical ${(n^{6},p+5.4)}$-crowd, as it is illustrated in
Figure~\ref{gadget-fig-3}.
\end{itemize}

\begin{figure}[tbh]
\centering
\subfigure[]{ \label{gadget-fig-1}
\includegraphics[scale=0.612]{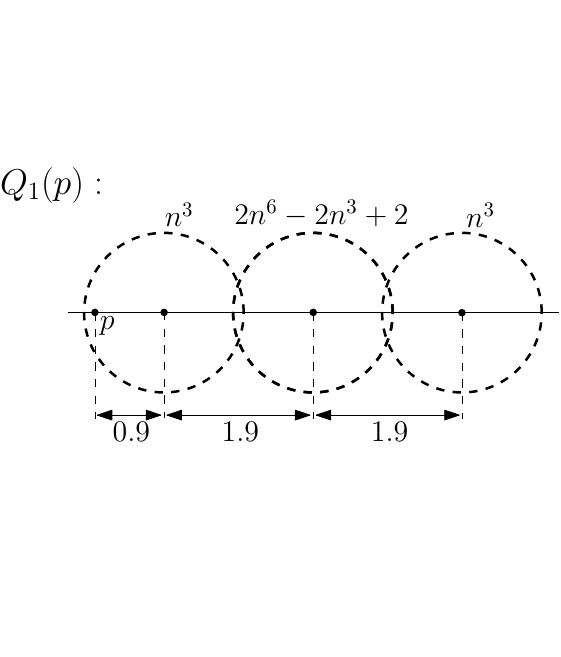}} \hspace{0.2cm} 
\subfigure[]{ \label{gadget-fig-2}
\includegraphics[scale=0.612]{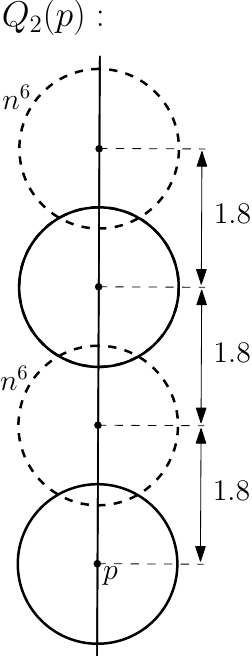}} \hspace{0.2cm} 
\subfigure[]{ \label{gadget-fig-3}
\includegraphics[scale=0.612]{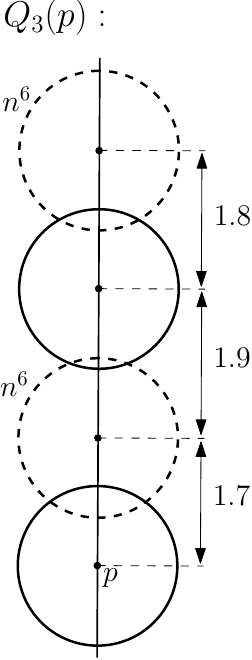}} 
\caption{The unit disk representations $Q_{1}(p)$, $Q_{2}(p)$, and $Q_{3}(p)$%
, where $p$ is a point on one of the tracks $T_{i,j}$, where $1\leq i\leq n$
and $j\in \{0,1\}$.}
\label{gadget-fig}
\end{figure}

In the above definition of the unit disk representation $Q_{k}(p)$, where~$%
k\in \{1,2,3\}$, the point $p$ is called the \emph{origin} of $Q_{k}(p)$.
Note that the origin $p$ of the representation $Q_{2}(p)$ (resp.~$Q_{3}(p)$)
is the center of a unit disk in $Q_{2}(p)$ (resp.~$Q_{3}(p)$). In contrast,
the origin $p$ of the representation $Q_{1}(p)$ is not the center of any
unit disk of $Q_{1}(p)$, however $p$ lies in $Q_{1}(p)$ within the area of
each of the $n^{3}$ unit disks of the horizontal ${(n^{3},p+0.9)}$-crowd of $%
Q_{1}(p)$. For every point $p$, each of $Q_{1}(p)$, $Q_{2}(p)$, and $%
Q_{3}(p) $ has in total $2n^{6}+2$ unit disks (cf.~Figure~\ref{gadget-fig}).

Furthermore, for any $i\in \{1,2,3\}$ and any two points $p$ and $p^{\prime
} $ in the plane, the unit disk representation $Q_{i}(p^{\prime })$ is an
isomorphic copy of the representation $Q_{i}(p)$, which is placed at the
origin $p^{\prime }$ instead of the origin $p$. Moreover, for any point $p$
in the vertical halfline of a track $T_{i,j}$, the unit disk representations 
$Q_{2}(p)$ and $Q_{3}(p)$ are almost identical: their only difference is
that the vertical ${(n^{6},p+1.8)}$-crowd in $Q_{2}(p)$ is replaced by the
vertical ${(n^{6},p+1.7)}$-crowd in $Q_{3}(p)$, i.e.~this whole crowd is
just moved downwards by $0.1$ in $Q_{3}(p)$.

\begin{observation}
\label{Qk(p)-obs}Let $k\in \{1,2,3\}$ and $p\in T_{i,j}$, where $i\in
\{1,2,\ldots ,n\}$ and $j\in \{0,1\}$. For every two adjacent vertices $u,v$
in the unit disk graph defined by $Q_{k}(p)$, $u$ and $v$ belong to a clique
of size at least $n^{6}+1$.
\end{observation}

\subsection{The unit disk representation $R_{n}$ of $G_{n}$\label%
{representation-Rn-subsec}}

We are now ready to iteratively construct the unit disk representation $%
R_{n} $ of the graph $G_{n}$, using the above gadgets $Q_{1}(p)$, $Q_{2}(p)$%
, and $Q_{3}(p)$, as follows:

\begin{itemize}
\item[(a)] for every $i\in \{1,2,\ldots ,n\}$ and for every $j\in \{0,1\}$,
add to~$R_{n}$:

\begin{itemize}
\item the gadget $Q_{1}(p)$, with its origin at the point $%
p=(0,(2(i-1)+j)d_{2})$, 
\end{itemize}

\item[(b)] for every $i\in \{1,2,\ldots ,n\}$, add to~$R_{n}$:

\begin{itemize}
\item the gadgets $Q_{1}(q_{i,0})$, $Q_{2}(r_{i,0})$, $Q_{3}(p_{i,0})$, and $%
Q_{3}(p_{i,1})$,

\item the gadgets $Q_{1}(p)$ and $Q_{1}(p^{\prime })$, with their origin at
the points ${p=(-d_{1},(2i-1)d_{2})}$ and $p^{\prime }=(-2d_{1},(2i-1)d_{2})$
of the track $T_{i,1}$, respectively, 
\end{itemize}

\item[(c)] for every ${i,k\in \{1,2,\ldots ,n\}}$ and for every ${j,\ell \in
\{0,1\}}$, where ${i\neq k}$, add to~$R_{n}$:

\begin{itemize}
\item the gadgets $Q_{1}(p)$ and $Q_{2}(p)$, with their origin at the
(unique) point $p$ that lies on the intersection of the tracks $T_{i,j}$ and 
$T_{k,\ell }$.
\end{itemize}
\end{itemize}

Similarly to Figure~\ref{track-fig}, we illustrate in Figure~\ref%
{gadgets-track-fig} the placement of the gadgets $Q_{1}(p)$, $Q_{2}(p)$, and 
$Q_{3}(p)$ in the unit disk representation $R_{n}$, for the various points $p
$ according to the above construction of $R_{n}$. In this figure, the
placement of a gadget $Q_{1}(p)$ (resp.~of a gadget $Q_{2}(p)$ and $Q_{3}(p)$) 
is depicted by a circled ``1'' (resp.~by a circled ``2'' and ``3'').

\begin{figure}[h]
\centering
\includegraphics[width=0.8\textwidth]{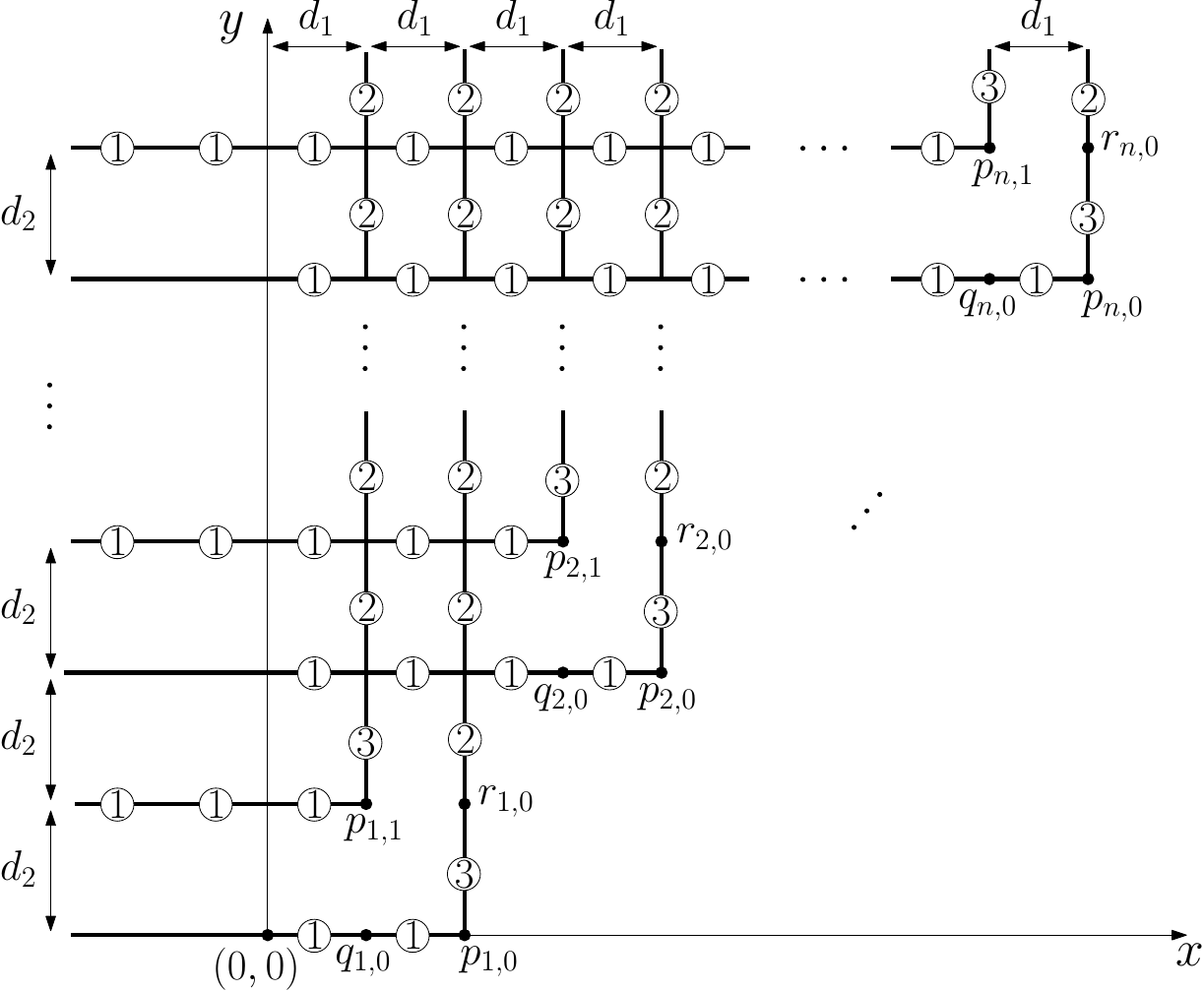}
\caption{The placement of the gadgets $Q_{1}(p)$, $Q_{2}(p)$, and $Q_{3}(p)$
in the unit disk representation $R_{n}$, for the various points $p$.}
\label{gadgets-track-fig}
\end{figure}

This completes the construction of the unit disk representation $R_{n}$ of
the graph $G_{n}=(V_{n},E_{n})$, in which the centers of all unit disks lie
on some track~$T_{i,j}$, where $i\in \{1,2,\ldots ,n\}$ and $j\in \{0,1\}$.

\begin{definition}
\label{Sij-def}Let $i\in \{1,2,\ldots ,n\}$ and $j\in \{0,1\}$. The vertex
set $S_{i,j}\subseteq V_{n}$ consists of all vertices of those copies of the
gadgets $Q_{1}(p)$, $Q_{2}(p)$, and $Q_{3}(p)$, whose origin $p$ belongs to
the track $T_{i,j}$.
\end{definition}

For every $v\in V_{n}$ let $c_{v}$ be the center of its unit disk in the
representation $R_{n}$. Note that, by Definition~\ref{Sij-def}, the unique
vertex ${v\in V_{n}}$, for which ${c_{v}\in T_{i,j}\cap T_{k,\ell }}$, where 
$i<k$ (i.e.~$c_{v}$ lies on the intersection of the vertical halfline of $%
T_{i,j}$ with the horizontal halfline of~${T_{k,\ell }}$), we have that~${%
v\in S_{i,j}}$. Furthermore note that $\{S_{i,j}:1\leq i\leq n,j\in
\{0,1\}\} $ is a partition of the vertex set $V_{n}$ of $G_{n}$. In the next
lemma we show that this is also a balanced $2n$-partition of $G_{n}$, i.e.~$%
|S_{i,j}|=|S_{k,\ell }|$ for every $i,k\in \{1,2,\ldots ,n\}$ and $j,\ell
\in \{0,1\}$.

\begin{lemma}
\label{number-vertices-track-lem}For every $i\in \{1,2,\ldots ,n\}$ and $%
j\in \{0,1\}$, we have that $|S_{i,j}|=4(n+1)(n^{6}+1)$.
\end{lemma}

\begin{proof}
Let first $j=0$. At part (a) of the above construction of the unit disk
representation $R_{n}$, the set $S_{i,0}$ receives the vertices of one copy
of the gadget $Q_{1}(p)$. At part (b) of the construction, $S_{i,0}$
receives the vertices of the gadgets $Q_{1}(q_{i,0})$, $Q_{2}(r_{i,0})$, and 
$Q_{3}(p_{i,0})$ (only due to the first bullet of part (b), since~${j=0}$).
Furthermore, at part (c) of the construction, $S_{i,0}$ receives the
vertices of~$2(i-1)$ copies of the gadget $Q_{1}(p)$ (i.e.~one for every
intersection of the horizontal halfline of $T_{i,0}$ with the vertical
halfline of a track $T_{k,\ell }$, where $k<i$) and the vertices of~$2(n-i)$
copies of the gadget $Q_{2}(p)$ (i.e.~one for every intersection of the
vertical halfline of $T_{i,0}$ with the horizontal halfline of a track $%
T_{k,\ell }$, where $k>i$). Note that this assignment of copies of the
gadgets $Q_{1}(p),Q_{2}(p),Q_{3}(p)$ to the vertices of~$S_{i,0}$ is
consistent with the definition of the partition of $V_{n}$ into $%
\{S_{i,j}\}_{i,j}$. Therefore, since each copy of the gadgets $%
Q_{1}(p),Q_{2}(p),Q_{3}(p)$ has $2n^{6}+2$ vertices, the set $S_{i,0}$ has
in total $(1+3+2(i-1)+2(n-i))\cdot (2n^{6}+2)=4(n+1)\cdot (n^{6}+1)$
vertices.

Let now $j=1$. At part (a) of the construction of $R_{n}$, the set $S_{i,1}$
receives similarly to the above the vertices of one copy of $Q_{1}(p)$. At
part (b) of the construction, $S_{i,1}$ receives the vertices of $%
Q_{3}(p_{i,1})$ (due to the first bullet) and the vertices of one copy of
each gadget $Q_{1}(p)$ and $Q_{1}(p^{\prime })$ (due to the second bullet).
Furthermore, at part (c) of the construction, $S_{i,1}$ receives similarly
to the above the vertices of $2(i-1)$ copies of the gadget $Q_{1}(p)$ and
the vertices of $2(n-i)$ copies of the gadget $Q_{2}(p)$. Note that this
assignment of copies of the gadgets $Q_{1}(p),Q_{2}(p),Q_{3}(p)$ to the
vertices of $S_{i,1}$ is again consistent with the definition of the
partition of $V_{n}$ into $\{S_{i,j}\}_{i,j}$. Therefore, since each copy of
the gadgets $Q_{1}(p),Q_{2}(p),Q_{3}(p)$ has $2n^{6}+2$ vertices, the set $%
S_{i,1}$ has in total $(1+1+2+2(i-1)+2(n-i))\cdot (2n^{6}+2)=4(n+1)\cdot
(n^{6}+1)$ vertices.\qed
\end{proof}

Consider the intersection point $p$ of two tracks $T_{i,j}$ and $T_{k,\ell }$%
, where $i\neq k$. Assume without loss of generality that $i<k$, i.e.~$p$
belongs to the vertical halfline of $T_{i,j}$ and on the horizontal halfline
of $T_{k,\ell }$, cf.~Figure~\ref{intersection-tracks-disks-fig-1}. Then $p$
is the origin of the gadget $Q_{2}(p)$ in the representation $R_{n}$ (cf.
part (c) of the construction of $R_{n}$). Therefore $p$ is the center of a
unit disk in $R_{n}$, i.e.~$p=c_{v}$ for some $v\in S_{i,j}\subseteq V_{n}$.
All unit disks of $R_{n}$ that intersect with the disk centered at point $p$
is shown in Figure~\ref{intersection-tracks-disks-fig-1}. Furthermore, the
induced subgraph $G_{n}[\{v\}\cup N(v)]$ on the vertices of $G_{n}$, which
correspond to these disks of Figure~\ref{intersection-tracks-disks-fig-1},
is shown in Figure~\ref{intersection-tracks-disks-fig-3}. In Figure~\ref%
{intersection-tracks-disks-fig-3} we denote by $K_{n^{6}}$ and $K_{n^{3}}$
the cliques with $n^{6}$ and with $n^{3}$ vertices, respectively, and the
thick edge connecting the two $K_{n^{3}}$'s depicts the fact that all
vertices of the two $K_{n^{3}}$'s are adjacent to each other.

Now consider a bend point $p_{i,j}$ of a variable $x_{i}$, where $j\in
\{0,1\}$. Then $p_{i,j}$ is the origin of the gadget $Q_{3}(p_{i,j})$ in the
representation $R_{n}$ (cf.~the first bullet of part (b) of the construction
of $R_{n}$). Therefore $p_{i,j}$ is the center of a unit disk in $R_{n}$,
i.e.~$p=c_{v}$ for some $v\in S_{i,j}\subseteq V_{n}$. All unit disks of $%
R_{n}$ that intersect with the disk centered at point $p_{i,j}$ are shown in
Figure~\ref{intersection-tracks-disks-fig-2}. Furthermore, the induced
subgraph $G_{n}[\{v\}\cup N(v)]$ of $G_{n}$ that corresponds to the disks of
Figure~\ref{intersection-tracks-disks-fig-2}, is shown in Figure~\ref%
{intersection-tracks-disks-fig-4}. In both Figures~\ref%
{intersection-tracks-disks-fig-1} and~\ref{intersection-tracks-disks-fig-2},
the area of the intersection of two crowds (i.e.~disks with dashed contour)
is shaded gray for better visibility.

\begin{figure}[tbh]
\centering
\subfigure[]{ \label{intersection-tracks-disks-fig-1}
\includegraphics[scale=0.5]{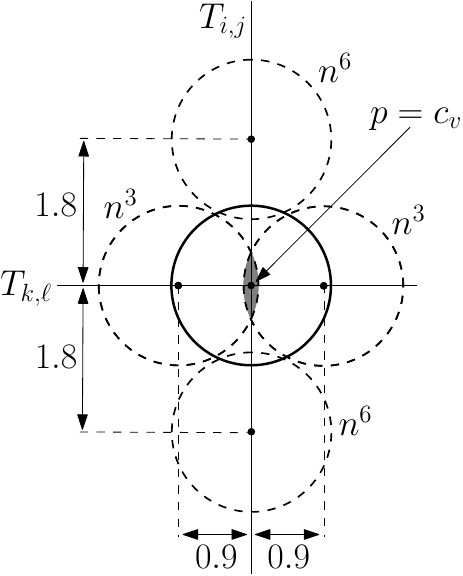}} \hspace{-0.38cm%
} 
\subfigure[]{ \label{intersection-tracks-disks-fig-2}
\includegraphics[scale=0.5]{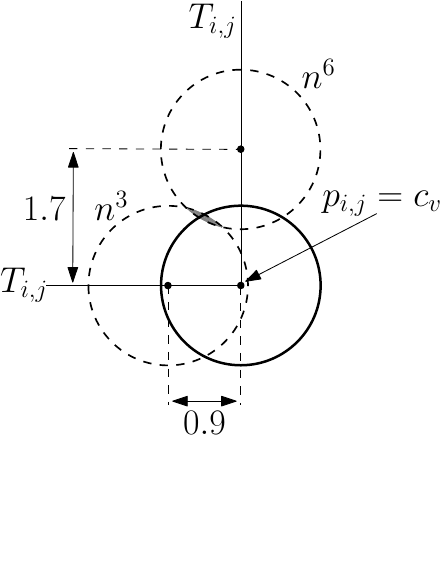}} \hspace{-0.9cm}
\subfigure[]{ \label{intersection-tracks-disks-fig-3}
\includegraphics[scale=0.512]{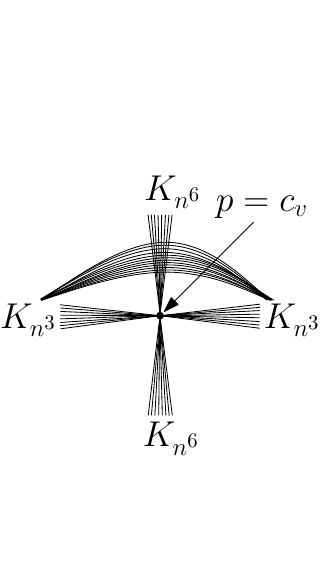}} \hspace{%
0.025cm} 
\subfigure[]{ \label{intersection-tracks-disks-fig-4}
\includegraphics[scale=0.512]{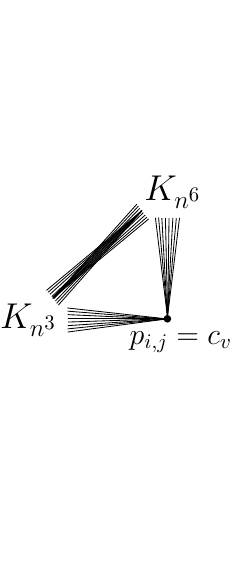}}
\caption{The disks in $R_{n}$ (a)~around the intersection point $p=c_{v}$ of
two tracks $T_{i,j}$ and $T_{k,\ell }$, where $i<k$, and (b)~around the bend
point $p_{i,j}=c_{v}$ of a variable $x_{i}$, where $j\in \{0,1\}$. (c)~The
induced subgraph of $G_{n}$ on the vertices of part~(a), and (d)~the induced
subgraph of $G_{n}$ for part~(b).}
\label{intersection-tracks-disks-fig}
\end{figure}

\begin{lemma}
\label{Sij-same-color}Consider an arbitrary bisection $\mathcal{B}$ of $%
G_{n} $ with size strictly less than~$n^{6}$. Then for every set $S_{i,j}$, $%
i\in \{1,2,\ldots ,n\}$ and $j\in \{0,1\}$, all vertices of $S_{i,j}$ belong
to the same color class of $\mathcal{B}$.
\end{lemma}

\begin{proof}
Let $i\in \{1,2,\ldots ,n\}$ and $j\in \{0,1\}$. Consider an arbitrary
gadget $Q_{k}(p)$ of $S_{i,j}$, where $k\in \{1,2,3\}$ and $p\in T_{i,j}$,
cf.~Definition~\ref{Sij-def}. Assume that there exist at least two vertices
of $Q_{k}(p)$ that belong to different color classes in the bisection $%
\mathcal{B}$ of $G_{n}$. Then, since the induced subgraph of $G_{n}$ on the
vertices of $Q_{k}(p)$ is connected, there exist at least two adjacent
vertices $u$ and $v$ in this subgraph that belong to two different color
classes of $\mathcal{B}$. Recall by Observation~\ref{Qk(p)-obs} that $u$ and 
$v$ belong to a clique $C$ of size at least $n^{6}+1$ in this subgraph.
Therefore, for each of the $n^{6}-1$ vertices $w\in C\setminus \{u,v\}$, the
edges $uw$ and $vw$ contribute exactly $1$ to the size of the bisection $%
\mathcal{B}$. Thus, since the edge $uv$ also contributes $1$ to the size of $%
\mathcal{B}$, it follows that the size of $\mathcal{B}$ is at least $n^{6}$,
which is a contradiction. Therefore for every gadget $Q_{k}(p)$ of $S_{i,j}$%
, where $k\in \{1,2,3\}$ and $p\in T_{i,j}$, all vertices of $Q_{k}(p)$
belong to the same color class of $\mathcal{B}$.

Now note that for all copies of the gadget $Q_{1}(p) $ in the set $S_{i,j}$,
their origins $p$ have the same $y$-coordinate (see Definition~\ref{Sij-def}%
). Similarly, for all copies of the gadgets $Q_{2}(p)$ and $Q_{3}(p)$ in $%
S_{i,j}$, their origins $p$ have the same $x$-coordinate. We order the
copies of $Q_{1}(p)$ in $S_{i,j}$ increasingly according to the $x$%
-coordinate of their origin $p$. Consider two consecutive copies of the
gadget $Q_{1}(p)$ in this ordering, with origins at points $p_{1}$ and $%
p_{2} $, respectively. Then, by the construction of the unit representation $%
R_{n}$ of $G_{n}$, the distance between $p_{1}$ and $p_{2}$ is equal to $%
d_{1}=5.6$. Therefore, it is easy to check that the vertices of the
horizontal ${(n^{3},p}_{1}{+4.7)}$-crowd of $Q_{1}(p_{1})$ and the vertices
of the horizontal ${(n^{3},p}_{2}{+0.9)}$-crowd of $Q_{1}(p_{2})$ induce a
clique of size $n^{3}+n^{3}=2n^{3}$ (cf.~Figure~\ref%
{intersection-tracks-disks-fig-1} and~\ref{intersection-tracks-disks-fig-3}%
). Thus, if the vertices of $Q_{1}(p_{1})$ belong to a different color class
than the vertices of $Q_{1}(p_{2})$, then $Q_{1}(p_{1})$ and $Q_{1}(p_{2})$
contribute $n^{6}$ to the size of the bisection $\mathcal{B}$, which is a
contradiction. Therefore all vertices of $Q_{1}(p_{1})$ and of $Q_{1}(p_{2})$
belong to the same color class. Furthermore, since this holds for any two
consecutive copies of the gadget $Q_{1}(p)$ in $S_{i,j}$, it follows that
the vertices of all copies of $Q_{1}(p)$ in $S_{i,j}$ belong to the same
color class.

Similarly, we order the copies of the gadgets $Q_{k}(p)$ in $S_{i,j}$, where 
$k\in \{2,3\}$, increasingly according to the $y$-coordinate of their origin 
$p$. Consider two consecutive copies of these gadgets in this ordering, with
origins at points $p_{1}$ and $p_{2}$, respectively. Note that, by the
construction of the unit representation $R_{n}$ of $G_{n}$, the point $p_{2}$
is either (i) the auxiliary point $r_{i,0}$ of track $T_{i,0}$ or (ii) the
intersection of the track $T_{i,j}$ with another track $T_{k,\ell }$, where $%
i<k$. Furthermore note that the gadget with origin at $p_{2}$ is $%
Q_{2}(p_{2})$, while the gadget with origin at $p_{1}$ is either $%
Q_{2}(p_{1})$ or $Q_{3}(p_{1})$. Suppose that the gadget with origin at $%
p_{1}$ is $Q_{2}(p_{1})$ (resp.~$Q_{3}(p_{1})$). Note that the distance
between $p_{1}$ and $p_{2}$ is equal to $d_{2}=7.2$. Therefore, it is easy
to check that the vertices of the vertical ${(n^{6},p}_{1}{+5.4)}$-crowd of $%
Q_{2}(p_{1})$ (resp.~of $Q_{3}(p_{1})$) and the single unit disk of $%
Q_{2}(p_{2})$, which is centered at point $p_{2}$, induce a clique of size $%
n^{6}+1$ (cf.~Figure~\ref{intersection-tracks-disks-fig-1} and~\ref%
{intersection-tracks-disks-fig-3}). Thus, if the vertices of $Q_{2}(p_{1})$
(resp.~of $Q_{3}(p_{1})$) belong to a different color class than the
vertices of $Q_{2}(p_{2})$, then $Q_{2}(p_{1})$ (resp.~$Q_{3}(p_{1})$) and $%
Q_{2}(p_{2})$ contribute $n^{6}$ to the size of the bisection $\mathcal{B}$,
which is a contradiction. Therefore all vertices of $Q_{2}(p_{1})$ (resp.~of 
$Q_{3}(p_{1})$) and of $Q_{2}(p_{2})$ belong to the same color class, and
thus the vertices of all copies of $Q_{k}(p)$ in $S_{i,j}$, where $k\in
\{2,3\}$, belong to the same color class.

It remains to prove that the vertices of the gadgets $Q_{1}(p)$ in $S_{i,j}$
belong to the same color class with the vertices of the gadgets $Q_{k}(p)$
in $S_{i,j}$, where $k\in \{2,3\}$. To this end, consider the rightmost
gadget $Q_{1}(p_{1})$ and the lowermost gadget $Q_{3}(p_{2})$ of $S_{i,j}$.
Note that, by the construction of the unit representation $R_{n}$ of $G_{n}$%
, the point $p_{2}$ is the bend point $p_{i,j}$ of the variable $x_{i}$ (cf.
Figure~\ref{intersection-tracks-disks-fig-2} and~\ref%
{intersection-tracks-disks-fig-4}). It is easy to check that the vertices of
the horizontal ${(n^{3},p}_{1}{+4.7)}$-crowd of $Q_{1}(p_{1})$ and the
vertical ${(n^{6},p}_{2}{+1.7)}$-crowd of $Q_{3}(p_{2})$ induce a clique of
size $n^{6}+n^{3}$ (cf.~Figure~\ref{intersection-tracks-disks-fig-1} and~\ref%
{intersection-tracks-disks-fig-3}). Thus, if the vertices of $Q_{1}(p_{1})$
belong to a different color class than the vertices of $Q_{3}(p_{2})$, then $%
Q_{1}(p_{1})$ and $Q_{3}(p_{2})$ contribute $n^{3}\cdot n^{6}=n^{9}$ to the
size of the bisection $\mathcal{B}$, which is a contradiction. Thus all
vertices of $Q_{1}(p_{1})$ and of $Q_{3}(p_{2})$ belong to the same color
class. Therefore, all vertices of $S_{i,j}$ belong to the same color class.%
\qed
\end{proof}

\section{Minimum bisection on unit disk graphs\label{H-phi-sec}}

In this section we provide our polynomial-time reduction from the monotone
Max-XOR($3$) problem to the minimum bisection problem on unit disk graphs.
To this end, given a monotone XOR($3$) formula $\phi $ with $n$ variables
and $m=\frac{3n}{2}$ clauses, we appropriately modify the auxiliary unit
disk graph $G_{n}$ of Section~\ref{Gn-sec} to obtain the unit disk graph $%
H_{\phi }$. Then we prove that the truth assignments that satisfy the
maximum number of clauses in $\phi $ correspond bijectively to the minimum
bisections in $H_{\phi }$.

We construct the unit disk graph $H_{\phi }=(V_{\phi },E_{\phi })$ from $%
G_{n}=(V_{n},E_{n})$ as follows. Let $(x_{i}\oplus x_{k})$ be a clause of $%
\phi $, where $i<k$. Let $p_{0}$ (resp.~$p_{1}$) be the unique point in the
unit disk representation $R_{n}$ that lies on the intersection of the tracks 
$T_{i,0}$ and $T_{k,1}$ (resp.~on the intersection of the tracks $T_{i,1}$
and $T_{k,0}$). For every point $p\in \{p_{0},p_{1}\}$, where we denote $%
p=(p_{x},p_{y})$, we modify the gadgets $Q_{1}(p)$ and $Q_{2}(p)$ in the
representation $R_{n}$ as follows:

\begin{itemize}
\item[(a)] replace the horizontal ${(n^{3},p+0.9)}$-crowd of $Q_{1}(p)$ by
the horizontal ${(n^{3}-1,p+0.9)}$-crowd and a single unit disk centered at $%
(p_{x}+0.9,p_{y}+0.02)$,

\item[(b)] replace the vertical ${(n^{6},p+1.8)}$-crowd of $Q_{2}(p)$ by the
vertical ${(n^{6}-1,p+1.8)}$-crowd and a single unit disk centered at $%
(p_{x}+0.02,p_{y}+1.8)$.
\end{itemize}

That is, for every point $p\in \{p_{0},p_{1}\}$, we first move one
(arbitrary) unit disk of the horizontal ${(n^{3},p+0.9)}$-crowd of $Q_{1}(p)$
upwards by $0.02$, and then we move one (arbitrary) unit disk of the
vertical ${(n^{6},p+1.8)}$-crowd of $Q_{2}(p)$ to the right by $0.02$. In
the resulting unit disk representation these two unit disks intersect,
whereas they do not intersect in the representation $R_{n}$. Furthermore it
is easy to check that for any other pair of unit disks, these disks
intersect in the resulting representation if and only if they intersect in $%
R_{n}$. The above modifications of $R_{n}$ for the clause $(x_{i}\oplus
x_{k})$ of $\phi $ are illustrated in Figure~\ref{Rn-modific-fig}.

\begin{figure}[tbh]
\centering
\subfigure[]{ \label{Rn-modific-fig-1}
\includegraphics[scale=0.6]{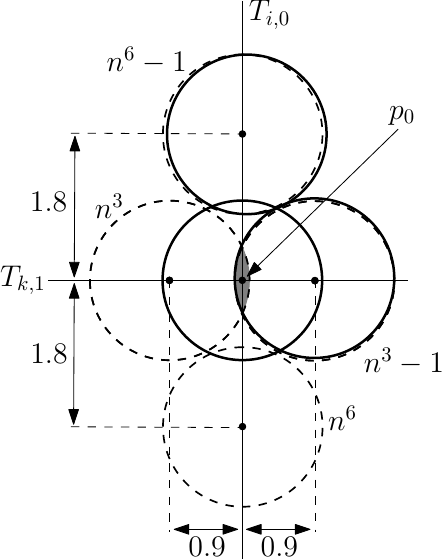}} \hspace{0.3cm} 
\subfigure[]{ \label{Rn-modific-fig-2}
\includegraphics[scale=0.6]{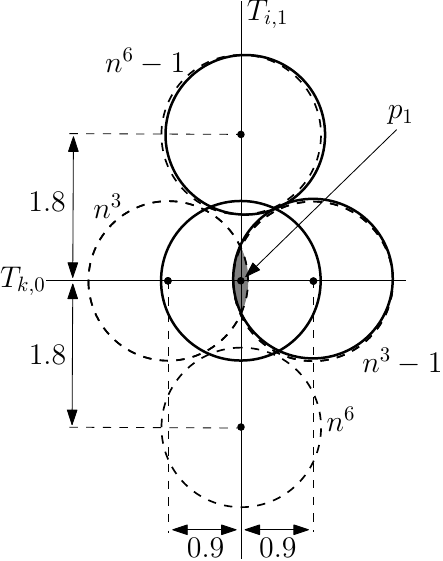}}
\caption{The modifications of the unit disk representation $R_{n}$ for the
clause $(x_{i}\oplus x_{k})$ of~$\protect\phi $, where (a)~$p_{0}$~is the
intersection of the tracks $T_{i,0}$ and $T_{k,1}$ and (b)~$p_{1}$~is the
intersection of the tracks $T_{i,1}$ and $T_{k,0}$. In both cases, one unit
disk of $Q_{1}(p)$ is moved upwards by~$0.02$ and one unit disk of $Q_{2}(p)$
is moved to the right by $0.02$, where $p\in \{p_{1},p_{2}\}$.}
\label{Rn-modific-fig}
\end{figure}

Denote by $R_{\phi }$ the unit disk representation that is obtained from $%
R_{n}$ by performing the above modifications for all clauses of the formula $%
\phi $. Then $H_{\phi }$ is the unit disk graph induced by $R_{\phi }$. Note
that, by construction, the graphs $H_{\phi }$ and $G_{n}$ have exactly the
same vertex set, i.e.~$V_{\phi }=V_{n}$, and that $E_{n}\subset E_{\phi }$.
In particular, note that the sets $S_{i,j}$ (cf.~Definition~\ref{Sij-def})
induce the same subgraphs in both $H_{\phi }$ and $G_{n}$, and thus the next
corollary follows directly by Lemma~\ref{Sij-same-color}.

\begin{corollary}
\label{Sij-same-color-H-phi-cor}Consider an arbitrary bisection $\mathcal{B}$
of $H_{\phi }$ with size strictly less than~$n^{6}$. Then for every set $%
S_{i,j}$, $i\in \{1,2,\ldots ,n\}$ and $j\in \{0,1\}$, all vertices of~$%
S_{i,j}$ belong to the same color class of $\mathcal{B}$.
\end{corollary}

\begin{theorem}
\label{H-phi-reduction-thm}There exists a truth assignment $\tau $ of the
formula $\phi $ that satisfies at least $k$ clauses if and only if the unit
disk graph $H_{\phi }$ has a bisection with value at most $2n^{4}(n-1)+3n-2k$.
\end{theorem}

\begin{proof}
($\Rightarrow $) Assume that the truth assignment $\tau $ of the variables $%
x_{1},x_{2},\ldots ,x_{n}$ of the formula $\phi $ satisfies at least $k$
clauses of $\phi $. We construct from the assignment $\tau $ a bisection $%
\mathcal{B}$ of the unit disk graph $H_{\phi }$ as follows. Denote the two
color classes of $\mathcal{B}$ by blue and red, respectively. For every
variable $x_{i}$, if $x_{i}=0$ in $\tau $ then we color all vertices of $%
H_{\phi }$ of the set $S_{i,0}$ blue and all vertices of the set $S_{i,1}$
red. Otherwise, if $x_{i}=1$ in $\tau $, we color all vertices of $H_{\phi }$
of the set $S_{i,0}$ red and all vertices of the set $S_{i,1}$ blue.
Therefore it follows by Lemma~\ref{number-vertices-track-lem} that for every
variable $x_{i}$, $1\leq i\leq n$, we have the same number of blue and red
vertices in $\mathcal{B}$, and thus $\mathcal{B}$ is indeed a bisection of
the graph $H_{\phi }$.

Recall that, in the formula $\phi $, every variable appears in exactly $3$
clauses, since $\phi $ is a monotone Max-XOR($3$) formula. Therefore $\phi $
has $m=\frac{3n}{2}$ clauses. Let now $1\leq i<k\leq n$. If $(x_{i}\oplus
x_{k})$ is not a clause of $\phi $ then, regardless of the value of $x_{i}$
and $x_{k}$ in the assignment $\tau $, the intersection of the tracks $%
T_{i,0},T_{i,1}$ with the tracks $T_{k,0},T_{k,1}$ contribute (due to the
construction of the graphs $G_{n}$ and $H_{\phi }$) exactly $%
2n^{3}+2n^{3}=4n^{3}$ edges to the value of the bisection $\mathcal{B}$.

If $(x_{i}\oplus x_{k})$ is one of the $k$ clauses of $\phi $ that are
satisfied by $\tau $, then the intersection of the tracks $T_{i,0},T_{i,1}$
with the tracks $T_{k,0},T_{k,1}$ contribute again $2n^{3}+2n^{3}=4n^{3}$
edges to the value of $\mathcal{B}$. Finally, if $(x_{i}\oplus x_{k})$ is
one of the $m-k$ clauses of $\phi $ that are not satisfied by $\tau $, then
the intersection of the tracks $T_{i,0},T_{i,1}$ with the tracks $%
T_{k,0},T_{k,1}$ contribute $\left( 2n^{3}+1\right) +\left( 2n^{3}+1\right)
=4n^{3}+2$ edges to the value of $\mathcal{B}$. Here the two additive
factors of ``+1'' are obtained due to the shifted and differently colored
disks in the construction. Summarizing, since $\tau$ satisfies at least $k$
clauses of $\phi $, the value of this bisection $\mathcal{B}$ of $H_{\phi}$
equals at most 
\begin{eqnarray*}
\left( {\binom{n}{2}}-m\right) 4n^{3}+(m-k)\left( 4n^{3}+2\right) +k\cdot
4n^{3} &=&2n^{4}(n-1)+2m-2k \\
&=&2n^{4}(n-1)+3n-2k
\end{eqnarray*}

($\Leftarrow $) Assume that $H_{\phi }$ has a minimum bisection $\mathcal{B}$
with value at most $2n^{4}(n-1)+3n-2k$. Denote the two color classes of $%
\mathcal{B}$ by blue and red, respectively. Since the size of $\mathcal{B}$
is strictly less than $n^{6}$, Corollary~\ref{Sij-same-color-H-phi-cor}
implies that for every $i\in \{1,2,\ldots ,n\}$ and $j\in \{0,1\}$, all
vertices of~the set $S_{i,j}$ belong to the same color class of $\mathcal{B}$%
. Therefore, all cut edges of $\mathcal{B}$ have one endpoint in a set $%
S_{i,j}$ and the other endpoint in a set $S_{k,\ell }$, where $(i,j)\neq
(k,\ell )$. Furthermore, since $\mathcal{B}$ is a bisection of $H_{\phi }$,
Lemma~\ref{number-vertices-track-lem} implies that exactly $n$ of the sets $%
\{S_{i,j}:1\leq i\leq n,j\in \{0,1\}\}$ are colored blue and the other $n$
ones are colored red in $\mathcal{B}$.

First we will prove that, for every $i\in \{1,2,\ldots ,n\}$, the sets $%
S_{i,0}$ and $S_{i,1}$ belong to different color classes in $\mathcal{B}$.
To this end, let $t\geq 0$ be the number of variables $x_{i}$, $1\leq i\leq
n $, for which both sets $S_{i,0}$ and $S_{i,1}$ are colored blue (such
variables $x_{i}$ are called \emph{blue}). Then, since $\mathcal{B}$ is a
bisection of $H_{\phi }$, there must be also $t$ variables $x_{i}$, $1\leq
i\leq n$, for which both sets $S_{i,0}$ and $S_{i,1}$ are colored red (such
variables $x_{i}$ are called \emph{red}), whereas $n-2t$ variables $x_{i}$,
for which one of the sets $\{S_{i,0},S_{i,1}\}$ is colored blue and the
other one red (such variables $x_{i}$ are called \emph{balanced}). Using the
minimality of the bisection $\mathcal{B}$, we will prove that $t=0$.

Every cut edge of $\mathcal{B}$ occurs at the intersection of the tracks of
two variables $x_{i},x_{k}$, where either both $x_{i},x_{k}$ are balanced
variables, or one of them is a balanced and the other one is a blue or red
variable, or one of them is a blue and the other one is a red variable.
Furthermore recall by the construction of the graph $H_{\phi }$ from the
graph $G_{n}$ that every clause $(x_{i}\oplus x_{k})$ of the formula $\phi $
corresponds to an intersection of the tracks of the variables $x_{i}$ and $%
x_{k}$. Among the $m$ clauses of $\phi $, let $m_{1}$ of them correspond to
intersections of tracks of two balanced variables, $m_{2}$ of them
correspond to intersections of tracks of a balanced variable and a blue or
red variable, and $m_{3}$ of them correspond to intersections of tracks of a
blue variable and a red variable. Note that $m_{1}+m_{2}+m_{3}\leq m$.

Let $1\leq i<k\leq n$. In the following we distinguish the three cases of
the variables $x_{i},x_{k}$ that can cause a cut edge in the bisection $%
\mathcal{B}$.

\begin{itemize}
\item $x_{i}$\textbf{\ and }$x_{k}$\textbf{\ are both balanced variables:}
in total there are $\frac{(n-2t)(n-2t-1)}{2}$ such pairs of variables, where
exactly $m_{1}$ of them correspond to a clause $(x_{i}\oplus x_{k})$ of the
formula $\phi $. It is easy to check that, for every such pair $x_{i},x_{k}$
that does not correspond to a clause of $\phi $, the intersection of the
tracks of $x_{i}$ and $x_{k}$ contributes exactly $2n^{3}+2n^{3}=4n^{3}$
edges to the value of $\mathcal{B}$. Furthermore, for each of the $m_{1}$
other pairs $x_{i},x_{k}$ that correspond to a clause of $\phi $, the
intersection of the tracks of $x_{i}$ and $x_{k}$ contributes either $4n^{3}$
or $4n^{3}+2$ edges to the value of $\mathcal{B}$. In particular, if the
vertices of the sets $S_{i,0}$ and $S_{k,1}$ have the same color in $%
\mathcal{B}$ then the pair $x_{i},x_{k}$ contributes $4n^{3}$ edges to the
value of $\mathcal{B}$, otherwise it contributes $4n^{3}+2$ edges. Among
these $m_{1}$ clauses, let $m_{1}^{\ast }$ of them contribute $4n^{3}$ edges
each and the remaining $m_{1}-m_{1}^{\ast }$ of them contribute $4n^{3}+2$
edges each.

\item \textbf{one of }$x_{i},x_{k}$\textbf{\ is a balanced variable and the
other one is a blue or red variable:} in total there are $(n-2t)2t$ such
pairs of variables, where exactly $m_{2}$ of them correspond to a clause $%
(x_{i}\oplus x_{k})$ of the formula $\phi $. It is easy to check that, for
every such pair $x_{i},x_{k}$ that does not correspond to a clause of $\phi $%
, the intersection of the tracks of $x_{i}$ and $x_{k}$ contributes exactly $%
2n^{3}+2n^{3}=4n^{3}$ edges to the value of $\mathcal{B}$. Furthermore, for
each of the $m_{2}$ other pairs $x_{i},x_{k}$ that correspond to a clause of 
$\phi $, the intersection of the tracks of $x_{i}$ and $x_{k}$ contributes $%
4n^{3}+1$ edges to the value of $\mathcal{B}$.

\item \textbf{one of }$x_{i},x_{k}$\textbf{\ is a blue variable and the
other one is a red variable:} in total there are $t^{2}$ such pairs of
variables, where exactly $m_{3}$ of them correspond to a clause $%
(x_{i}\oplus x_{k})$ of the formula $\phi $. It is easy to check that, for
every such pair $x_{i},x_{k}$ that does not correspond to a clause of $\phi $%
, the intersection of the tracks of $x_{i}$ and $x_{k}$ contributes exactly $%
4\cdot 2n^{3}=8n^{3}$ edges to the value of $\mathcal{B}$. Furthermore, for
each of the $m_{3}$ other pairs $x_{i},x_{k}$ that correspond to a clause of 
$\phi $, the intersection of the tracks of $x_{i}$ and $x_{k}$ contributes $%
8n^{3}+2$ edges to the value of $\mathcal{B}$.
\end{itemize}

Therefore, the value of $\mathcal{B}$ can be computed as follows:%
\begin{eqnarray}
&&\left( \frac{(n-2t)(n-2t-1)}{2}-m_{1}\right) 4n^{3}+m_{1}^{\ast
}4n^{3}+(m_{1}-m_{1}^{\ast })(4n^{3}+2)  \notag \\
&&+\left( (n-2t)2t-m_{2}\right) 4n^{3}+m_{2}\left( 4n^{3}+1\right)  \notag \\
&&+\left( t^{2}-m_{3}\right) 8n^{3}+m_{3}\left( 8n^{3}+2\right)  \notag \\
&=&(n-2t)(n-2t-1)2n^{3}+2(m_{1}-m_{1}^{\ast })  \label{value-B-eq} \\
&&+(n-2t)2t4n^{3}+m_{2}  \notag \\
&&+t^{2}\cdot 8n^{3}+2m_{3}  \notag \\
&=&(n-2t)(n+2t)2n^{3}-(n-2t)2n^{3}+t^{2}8n^{3}+2(m_{1}-m_{1}^{\ast
})+m_{2}+2m_{3}  \notag \\
&=&2n^{4}\left( n-1\right) +4n^{3}t+2(m_{1}-m_{1}^{\ast })+m_{2}+2m_{3} 
\notag
\end{eqnarray}

Note now that $0\leq 2(m_{1}-m_{1}^{\ast })+m_{2}+2m_{3}\leq 2m=3n<4n^{3}$.
Therefore, since the value of the bisection $\mathcal{B}$ (given in (\ref%
{value-B-eq})) is minimum by assumption, it follows that $t=0$. Thus for
every $i\in \{1,2,\ldots ,n\}$ the variable $x_{i}$ of $\phi $ is balanced
in the bisection $\mathcal{B}$, i.e.~the sets $S_{i,0}$ and $S_{i,1}$ belong
to different color classes in $\mathcal{B}$. That is, $m_{1}=m$ and $%
m_{2}=m_{3}=0$, and thus the value of $\mathcal{B}$ is by (\ref{value-B-eq})
equal to $2n^{4}\left( n-1\right) +2(m-m_{1}^{\ast })$. On the other hand,
since the value of $\mathcal{B}$ is at most $2n^{4}(n-1)+3n-2k$ by
assumption, it follows that $2(m-m_{1}^{\ast })\leq 3n-2k$. Therefore, since 
$m=\frac{3n}{2}$, it follows that $m_{1}^{\ast }\geq k$.

We define now from $\mathcal{B}$ the truth assignment $\tau $ of $\phi $ as
follows. For every $i\in \{1,2,\ldots ,n\}$, if the vertices of the set $%
S_{i,0}$ are blue and the vertices of the set $S_{i,1}$ are red in $\mathcal{%
B}$, then we set $x_{i}=0$ in $\tau $. Otherwise, if the vertices of the set 
$S_{i,0}$ are red and the vertices of the set $S_{i,1}$ are blue in $%
\mathcal{B}$, then we set $x_{i}=1$ in $\tau $. Recall that $m_{1}^{\ast }$
is the number of clauses of $\phi $ that contribute $4n^{3}$ edges each to
the value of $\mathcal{B}$, while the remaining $m-m_{1}^{\ast }$ clauses of 
$\phi $ contribute $4n^{3}+2$ edges each to the value of $\mathcal{B}$.
Thus, by the construction of $H_{\phi }$ from $G_{n}$, for every clause $%
(x_{i}\oplus x_{k})$ of $\phi $ that contributes $4n^{3}$ (resp.~$4n^{3}+2$)
to the value of $\mathcal{B}$, the vertices of the sets $S_{i,0}$ and $%
S_{k,1}$ have the same color (resp.~$S_{i,0}$ and $S_{k,1}$ have different
colors) in $\mathcal{B}$. Therefore, by definition of the truth assignment $%
\tau $, there are exactly $m_{1}^{\ast }$ clauses $(x_{i}\oplus x_{k})$ of $%
\phi $ where $x_{i}\neq x_{k}$ in $\tau $, and there are exactly $%
m-m_{1}^{\ast }$ clauses $(x_{i}\oplus x_{k})$ of $\phi $ where $x_{i}=x_{k}$
in $\tau $. That is, $\tau $ satisfies exactly $m_{1}^{\ast }\geq k$ of the $%
m$ clauses of $\phi $. This completes the proof of the theorem.\qed
\end{proof}

We can now state our main result, which follows by Theorem~\ref%
{H-phi-reduction-thm} and Lemma~\ref{xor-np-hard-lem}.

\begin{theorem}
\label{min-bisection-unit-disk-np-hard-thm}\textsc{Min-Bisection} is NP-hard
on unit disk graphs.
\end{theorem}

\section{Concluding Remarks\label{conclusions}}

In this paper we proved that \textsc{Min-Bisection} is NP-hard on unit disk
graphs by providing a polynomial time reduction from the monotone Max-XOR($3$%
) problem, thus solving a longstanding open question. As pointed out in the
Introduction, our results indicate that \textsc{Min-Bisection} is probably
also NP-hard on planar graphs, or equivalently on grid graphs with an
arbitrary number of holes, which remains yet to be proved. Our construction
for the NP-hardness reduction involved huge cliques, and thus it seems that
a different approach would be needed to possibly prove NP-hardness of 
\textsc{Min-Bisection} for planar graphs.


\begin{thebibliography}{10}

\bibitem{AkyildizSSC02}
I.~Akyildiz, W.~Su, Y.~Sankarasubramaniam, and E.~Cayirci.
\newblock Wireless sensor networks: {A} survey.
\newblock {\em Computer Networks}, 38:393--422, 2002.

\bibitem{GraphPartitioning-Book}
C.-E. Bichot and P.~Siarry, editors.
\newblock {\em Graph Partitioning}.
\newblock Wiley, 2011.

\bibitem{BradonjicEFSS10}
M.~Bradonjic, R.~Els{\"a}sser, T.~Friedrich, T.~Sauerwald, and A.~Stauffer.
\newblock Efficient broadcast on random geometric graphs.
\newblock In {\em Proceedings of the 21st annual {ACM-SIAM Symposium on
  Discrete Algorithms} {(SODA)}}, pages 1412--1421, 2010.

\bibitem{Breu98}
H.~Breu and D.~G. Kirkpatrick.
\newblock Unit disk graph recognition is {NP}-hard.
\newblock {\em Computational Geometry}, 9(1-2):3--24, 1998.

\bibitem{Bui87}
T.~Bui, S.~Chaudhuri, T.~Leighton, and M.Sipser.
\newblock Graph bisection algorithms with good average case behavior.
\newblock {\em Combinatorica}, 7:1987, 171--191.

\bibitem{CyganLPPS14}
M.~Cygan, D.~Lokshtanov, M.~Pilipczuk, M.~Pilipczuk, and S.~Saurabh.
\newblock Minimum bisection is fixed parameter tractable.
\newblock In {\em Proceedings of the 46th Annual Symposium on the Theory of
  Computing (STOC)}, 2014.
\newblock To appear.

\bibitem{DellingGRW12}
D.~Delling, A.~V. Goldberg, I.~Razenshteyn, and R.~F. Werneck.
\newblock Exact combinatorial branch-and-bound for graph bisection.
\newblock In {\em Proceedings of the 14th Meeting on Algorithm Engineering {\&}
  Experiments (ALENEX)}, pages 30--44, 2012.

\bibitem{DiazPPS01}
J.~D\'{\i}az, M.~D. Penrose, J.~Petit, and M.~J. Serna.
\newblock Approximating layout problems on random geometric graphs.
\newblock {\em Journal of Algorithms}, 39(1):78--116, 2001.

\bibitem{DiazPS02}
J.~D\'{\i}az, J.~Petit, and M.~Serna.
\newblock A survey on graph layout problems.
\newblock {\em ACM Computing Surveys}, 34:313--356, 2002.

\bibitem{FeldmannW11}
A.~E. Feldmann and P.~Widmayer.
\newblock An ${O}(n^4)$ time algorithm to compute the bisection width of solid
  grid graphs.
\newblock In {\em Proceedings of the 19th annual European Symposium on
  Algorithms (ESA)}, pages 143--154, 2011.

\bibitem{GareyJohnson}
M.~R. Garey and D.~S. Johnson.
\newblock {\em Computers and intractability: {A} guide to the theory of
  {NP}-completeness}.
\newblock W. H. Freeman \& Co., 1979.

\bibitem{Hromkovic05}
J.~Hromkovic, R.~Klasing, A.~Pelc, P.~Ruzicka, and W.~Unger.
\newblock {\em Dissemination of Information in Communication Networks -
  Broadcasting, Gossiping, Leader Election, and Fault-Tolerance}.
\newblock Texts in Theoretical Computer Science. An EATCS Series. Springer,
  2005.

\bibitem{JansenKLS05}
K.~Jansen, M.~Karpinski, A.~Lingas, and E.~Seidel.
\newblock Polynomial time approximation schemes for {Max-Bisection} on planar
  and geometric graphs.
\newblock {\em SIAM Journal on Computing}, 35(1):110--119, 2005.

\bibitem{Kahruman09}
S.~Kahruman-Anderoglu.
\newblock {\em Optimization in geometric graphs: Complexity and approximation}.
\newblock PhD thesis, Texas A \& M University, 2009.

\bibitem{Karpinski02}
M.~Karpinski.
\newblock Approximability of the minimum bisection problem: {An} algorithmic
  challenge.
\newblock In {\em Proceedings of the 27th International Symposium on
  Mathematical Foundations of Computer Science (MFCS)}, pages 59--67, 2002.

\bibitem{KL70}
B.~Kernighan and S.~Lin.
\newblock An efficient heuristic procedure for partitioning graphs.
\newblock {\em Bell System Technical Journal}, 49(2):291--307, 1970.

\bibitem{Kratochvil96}
J.~Kratochv{\'i}l.
\newblock Intersection graphs of noncrossing arc-connected sets in the plane.
\newblock In {\em Proceedings of the 4th Int. Symp. on Graph Drawing (GD)},
  pages 257--270, 1996.

\bibitem{MacGregor78}
R.~MacGregor.
\newblock {\em On partitioning a graph: {A} theoretical and empirical study}.
\newblock PhD thesis, University of California, Berkeley, 1978.

\bibitem{PapadimitriouSideri96}
C.~H. Papadimitriou and M.~Sideri.
\newblock The bisection width of grid graphs.
\newblock {\em Mathematical Systems Theory}, 29(2):97--110, 1996.

\bibitem{Racke08}
H.~R{\"a}cke.
\newblock Optimal hierarchical decompositions for congestion minimization in
  networks.
\newblock In {\em Proceedings of the 40th Annual ACM Symposium on Theory of
  Computing (STOC)}, pages 255--264, 2008.

\bibitem{Yannakakis78}
M.~Yannakakis.
\newblock Node-and edge-deletion {NP}-complete problems.
\newblock In {\em Proceedings of the 10th annual {ACM} {Symposium on Theory of
  Computing} (STOC)}, pages 253--264, 1978.

\end{thebibliography}
\end{document}